\documentclass[a4paper]{amsart}

\usepackage{amscd,amsthm,amsfonts,amssymb,amsmath}
\usepackage{fullpage}
\usepackage[all]{xy}

     \newcommand{\BF}{{\mathbb {F}}}

     \newcommand{\BZ}{{\mathbb {Z}}}

    \newcommand{\fa}{{\mathfrak{a}}}

    \newcommand{\Hom}{{\mathrm{Hom}}}
    
    \renewcommand{\Im}{{\mathrm{Im}}}

    \newcommand{\tr}{{\mathrm{tr}}}

    \theoremstyle{plain}
    \newtheorem{thm}{Theorem}[section] \newtheorem{cor}[thm]{Corollary}
    \newtheorem{lem}[thm]{Lemma}  \newtheorem{prop}[thm]{Proposition}
     \newtheorem{defn}[thm]{Definition}

\theoremstyle{remark} 
\theoremstyle{remark} 
\theoremstyle{remark} \newtheorem{example}{Example}

    \numberwithin{equation}{section}

\begin{document}

\title{ON the hulls of group codes}
\author{Xiheng Deng, Yuan Ren}

\address{X. Deng is with the School of Mathematical Sciences, University of Electronic Science and Technology of China,
Chengdu, 611731, China  (e-mail:
2022110801005@std.uestc.edu.cn).}
\address{Y. Ren is with the School of Mathematical Sciences, University of Electronic Science and Technology of China,
Chengdu, 611731, China  (e-mail:
ry@uestc.edu.com). }

\begin{abstract}
Let $\BF_q$ be a finite field and $G$ a finte group with $(|G|,q)=1$. By a group code in $\BF_q[G]$ we mean a two-sided ideal in $\BF_q[G]$. We will prove a general criterion for the existence of group codes with given hull dimension, and then apply it to deduce explicit criterions for existence of group codes with hull dimension $\leq3$. In particular our criterion for the existence of $1$-dimensional hulls generalizes that of \cite{LYMH} which consider only abelian groups $G$.
\end{abstract}

\maketitle


\section{Introduction}  

Let $\BF_q$ be a finite field. For a linear code $C$ over $\BF_q$, its hull is defined to be the intersection \(C \cap C^\perp\) where $C^{\perp}$ is the dual of $C$ with respect to the Euclidean inner product. The study of hulls in coding theory traces its roots to foundational works on cyclic codes, where the duality structure was elegantly described through polynomial ring factorizations (see \cite{Ber}). The hulls with small dimensions have been continuously valued due to their essential role in determining the complexity of algorithms for checking permutation equivalence of two linear codes (see \cite{Leon,Sendrier}) and in computing the automorphism group of a linear code \cite{Leon2}. Recently, Calderbank et al. highlight the role of hulls in quantum stabilizer codes and use self-orthogonal codes (a code is called self-orthogonal if its hull equals to itself) to consturct quantum error-correcting codes (QECCs) (see \cite{CR,CS}). Moreover, Brun et al. give the concept of entanglement-assisted quantum error-correcting codes (EAQECCs), which can be construct without requiring the code to be self-orthogonal \cite{Bowen,Brun}. An amount of study of EAQECCSs like \cite{LP,Sok,Sok2} then promote the theoretical research on the hull of linear codes. Another motivation comes from analyzing the security of the McEliece cryptosystem, exploring the role of the hull of codes in resisting structural attacks (such as support splitting attacks) (see \cite{Sen}). It points out that a large hull may obscure the algebraic structure of the code, thereby enhancing security.

In this paper we investigate hulls of group codes. Let $G$ be a finite group and let $\BF_q[G]$ be the group algebra. By a group code in $\BF_q[G]$ we mean a two-sided ideal in the $\BF_q$-algebra. See \S2.2 for more details about group algebra over fields. Assume that $|G|$ is prime to $q$. Then, any group code decomposes as a sum minimal ones, which are genenrated by primitive central idempotents. Moreover,  we will see that there is a bijection between primititve central idempotents in $\BF_q[G]$ and Galois orbits of isomorphism classes of irreducible representations of $G$ over $\bar{\BF}_q$. In particular, dual codes then corresponde to contragradient representations, which enables us to establsh a crieterion for existence of group codes with given hull dimensions. More precisely, let $Irr(G)$ be the set of isomorphism classes of irreducible reprsesentations of $G$ over $\overline{\BF}_q$. For any $\rho\in Irr(G)$ we denote by $d(\rho)$ (resp. $s(\rho)$) its dimension (resp. the cardinality of the Galois orbit $G_{\BF_q}\rho$), and define
\begin{align*}
    d[\rho]=s(\rho)d(\rho)^2.
\end{align*}
Then, our main result is the following theorem.

\begin{thm}\label{main1}
Let $\BF_q$ be a finite field and let $G$ be a finite group with $(|G|,q)=1$. Then, for any positive integer $k$, there exists a group code $C$ in $\BF_q[G]$ with $\dim_{\BF_q}h(C)=k$ if and only if there is a subset $S\subseteq [Irr(G)]$ satisfying the following conditions: 
\begin{enumerate}
\item $\sum\limits_{[\rho]\in S}d[\rho]=k$;
\item $[\rho^*]\notin S$ for any $[\rho]\in S$.
\end{enumerate}
\end{thm}

The proof of the above theorem will be given in \S3. From a practical standpoint, codes with small hulls are advantageous in quantum error correction (e.g., simplifying stabilizer measurements) and cryptographic protocols requiring efficient dual code manipulations. In \S3 we will give criterions for existence of group codes with hull dimensions $\leq3$. For example, the following theorem is the case about existence of $1$-dimensional hulls which generalizes previous result for ableian codes(see \cite{LYMH}).

\begin{thm}\label{main2}
Let $\BF_q$ be a finite field and $G$ a finite group with $(|G|, q) = 1$. Then there exists a group
code $C\subseteq\BF_q[G]$ with $\dim_{\BF_q}h(C) = 1$ if and only if $(exp(G^{ab}), q-1)\geq3$, where $G^{ab}$ is the abelianization
of $G$ and $exp(G^{ab})$ its exponent.
\end{thm}

The paper is organized as follows. In \S2 we recall some basic notions and results about semi-simple group algebras over general fields, emphasizing the role of central primitive idempotents. In \S3 we study semi-simple group algebras over finite fields and give the proof of Theorem~\ref{main1}. In \S4 we investigate group codes with hull dimension $\leq3$. 

\section{Preliminaries}
In this section we briefly review some backgrounds about semi-simple group algebras over fields. For more details we refer to \cite{Webb}.

\subsection{Semi-simple rings}
In this subsectioin we briefly review the theory of semi-simple rings. Let us first recall some basic notions. Let A be a unital ring. If $\fa$ is an additional subgroup of $A$ such that $xa, ax \in \fa$ for any $x \in A$ and $a \in \fa$, then we say that $\fa$ is a 2-sided ideal, or simply, an ideal. If $e \in A$ satisfies $e^2 = e$, then we call e as an idempotent. Moreover, if e is not a sum of two nonzero idempotents, then we say that e is a primitive idempotent. We denote by $Z(A)$ the center of A, which is the subring consisting of those $x \in A$ such that $xy=yx$ for any $y \in A$. If $e \in Z(A)$ is an idempotent (resp. a primitive idempotent), then we call it as a central idempotent (resp. a central primitive idempotent).
By an $A$-module we always mean a left module over $A$. If $M$ is an $A$-module and it has no submodule except $(0)$ and itself, then we say that $M$is simple. If M is an $A$-module and it is isomorphic to a direct sum of simple $A$-modules, then we say that $M$ is semi-simple. If every $A$-module is semi-simple, then we say that $A$ is a semi-simple ring.
\begin{prop}\label{idem}
  Let $A$ be a unital ring. Then there exists a bijection between decompositions 
  \begin{align*}
      A = A_1 \oplus A_2 \oplus \dots \oplus A_r
  \end{align*}
  as a direct sum of ideals $A_i$ with expressions
  \begin{align*}
      1 = e_1 + e_2 + \dots +e_r
  \end{align*}
  as a sum of orthogonal central idempotents, in such away that $e_i$ is the identity element of $A_i$ and $A_i$ is indecomposable as a ring if and only if  $e_i$ is primitive. If every $A_i$ is indecomposable as a ring, then the subsets $A_i$ and also the primitive central idempotents $e_i$ are uniquely determined; furthermore, every central idempotent can then be written as a sum of certain of these $e_i$.
\end{prop}
\begin{proof}
    This is Proposition 3.6.1 of \cite{Webb}.
\end{proof}
For any unital ring $A$, it has a left module structure over itself, which is denoted as $_AA$.
\begin{prop}\label{A_W}
    Let $A$ be a finite dimensional semi-simple algebra over a field $F$. If $_AA \cong \bigoplus \limits_{i=1}^r S_i^{n_i}$ where $S_1, \dots, S_r$ are nonisomorphic simple modules, then
    \begin{align*}
        A \cong \prod \limits_{i=1}^r M_{n_i}(D_i)
    \end{align*}
    where $D_i$ is a division algebra of finite degree over $F$ for each $1 \le i \le r$.We have that $S_1, \dots, S_r$ is a complete set of representations of  the isomorphism classes of simple $A$-modules. Furthermore, when F is algebraically closed, we have $D_i=F (\forall 1\le i \le r)$, $n_i = \dim_FS_i$ and $\dim_FA = n_1^2 + n_2^2 + \dots + n_r^2$.
\end{prop}
\begin{proof}
    This is Artin-Wedderburn theorem. See Theorem 2.1.3 and Corollary 2.1.4 of \cite{Webb}.
\end{proof}
\begin{cor}
    Let $A$ be a finite dimensional semi-simple algebra over a field $F$. Then, for any ideal $\fa$, there is a unique central idempotent $e \in A$ such that $\fa = (e)$. Such an ideal $\fa = (e)$ is minimal if and only if $e$ is primitive. 
\end{cor}
\begin{proof}
    By Proposition~\ref{A_W} we may identify $A \cong \prod \limits_{i=1}^r M_{n_i}(D_i)$ with $D_i$ a division algebra over $F$. It is well known that a ring of the  form $M_n(D)$, where $n$ is a positive integer and $D$ is a division ring, has no ideals other than $(0)$ and itself, which implies that it is indecomposable. In particular, it follows from Proposition~\ref{idem} that there are uniquely determined primitive central idempotents $e_1, \dots, e_r$ such that $M_{n_i}(D_i) = (e_i)$. Since any ideal $\fa$ of $A$ is equal to $\bigoplus \limits _{i\in I}M_{n_i}(D_i)$ for some subset $I \subseteq \{ 1, 2, \dots, n \}$, we have $\fa = (e)$ with $e = \sum \limits_{i \in I}e_i$ which is an idempotent as $e_i$'s are orthogonal. The uniqueness and the remaining assertion are clear.  
\end{proof}

\subsection{Semi-simple group algebras over fields}
In this subsection, let $F$ be a field and let $G$ be a finite group. Then the group algebra $F[G]$ consists of all formal sums $\sum \limits _{g \in G}a(g)g$ with addition and multiplication defined as follows: 
\begin{align*}
	\begin{cases}
		\left(\sum_{g\in G}f_1(g)g\right)+\left(\sum_{g\in G}f_2(g)g\right)=\sum_{g\in G}\left(f_1(g)+f_2(g)\right)g\\
		\left(\sum_{g\in G}f_1(g)g\right)\cdot\left(\sum_{g\in G}f_2(g)g\right)=\sum_{g\in G}\left(\sum_{h\in G}f_1(h)f_2(h^{-1}g)\right)g.
	\end{cases}
\end{align*}
Thus, $F[G]$ is a unital $F$-algebra which is commutative if and only if $G$ is abelian. From now on, we will assume that $\left(\lvert G \rvert,char(F)\right)=1$ so that $F[G]$ is a finite dimensional semi-simple algebra over $F$ by Maschke's theorem (see, e.g., Theorem 1.2.1 and Corollary 1.2.5 of \cite{Webb}). 
Fix an algebraic closure $\overline{F}$ of $F$. Recall that, if $\rho : G \rightarrow GL_n(\overline{F})$ is a group homomorphism with $n$ an integer $\ge 1$, then we say that $\rho$ is a representation of $G$ over $\overline{F}$ of degree $n$. If $\rho_1, \rho_2$ are two such representations and there exists a $T \in GL_n(\overline{F})$ such that $\rho_2 = T^{-1} \circ \rho_1 \circ T$, then we say that $\rho_1$ is isomorphic to $\rho_2$, denoted as $\rho_1 \cong \rho_2$. A sub-representation of a representation is a subspace stable under the action of $G$. If a representation has no sub-representation except $(0)$ and itself, then we say that it is irreducible. We denote by $$Irr(G)$$ the set of isomorphism classes of irreducible (finite dimensional) representation of $G$ over $\overline{F}$. For any representation $\rho : G \rightarrow GL_n(\overline{F})$, let $$\chi_{\rho} = \tr(\rho) : G \rightarrow \overline{F}$$ which is called the character of $G$. Note that $\chi_{\rho}$ depends only on the isomorphism class of $\rho$ For any two finite dimensional representations $\rho_1, \rho_2$, we define $$<\chi_{\rho_1}, \chi_{\rho_2}> = \frac{1}{|G|} \sum\limits_{g \in G} \chi_{\rho_1}(g) \chi_{\rho_2}(g^{-1}).$$
\begin{lem}\label{delta}
    If {$\rho_1, \rho_2, \dots, \rho_r$} is a full set of representatives for $Irr(G)$, then $$<\chi_{\rho_i}, \chi_{\rho_j}> = \delta_{ij}$$ where $\delta_{ij}$ is the Kronecker delta.
\end{lem}
\begin{proof}
    For any $1 \le i \le r$, let $\phi_i$ be the Brauer character of $\rho_i$ (see, e.g., page 193 of \cite{Webb}). Since $(|G|,p)=1$, every irreducible representation of $G$ over $\overline{F}$ is liftable by Fong-Swan-Rukolaine theorem (see, e,g, Theorem 9.4.12 of \cite{Webb}), and we have $\chi_{\rho_i}$ is the reduction of $\phi_i$ for $1\le i \le r$. Therefore, the assertion follows from the row orthogonality relations for Brauer characters(see Theorem 10.2.2 of \cite{Webb}). 
\end{proof}
\begin{defn}
    For any irreducible representation $\rho$ of $G$ over $\overline{F}$ with dimension $d$, we define $$e(\rho) = \frac{d}{|G|}\sum\limits_{g \in G}\chi(g^{-1})g.$$
\end{defn}
\begin{prop}
    Let $\{\rho_1, \rho_2, \dots, \rho_r\}$ be a full set of representatives for $Irr(G)$. Then there is an isomorphism of $\overline{F}$-algebra: $$\Psi:\overline{F}[G]\cong\prod\limits_{i=1}^rM_{d_i}(\overline{F})$$ with $d_i=\dim\rho_i\ (1\le i\le r)$, which maps any $\sum\limits_{g\in G}f(g)g$ to $(\dots, \sum\limits_{g\in G}f(g)\rho_i(g), \dots)_{1\le i\le r}$. Moreover, if $\epsilon_i=(\dots,0, I_{d_i\times d_i},0, \dots)\in\prod\limits_{i=1}^rM_{d_i}(\overline{F})$ for any $1\le i\le r$, then $\Psi(e(\rho_i))=\epsilon_i$. In particular, $\{e(\rho_i)|1\le i\le r\}$ is the set of all central primitive idempotents of $\overline{F}[G]$.
\end{prop}
\begin{proof}
    Let us first show that $\Psi(e(\rho_i))=\epsilon_i$ for any $1\le i\le r$. Indeed, the $j$-th component $T_j$ of $\Psi(e(\rho_i))$ equals $\frac{d_i}{|G|}\sum\limits_{g\in G}\chi_{\rho_i}(g^{-1})\rho_j(g)$. Since characters are class function, this operator commutes with $\rho_j$ and is hence a scalar by Schur's lemma (see Theorem 2.1.1 of Webb\cite{Webb}). Moreover, this scalar equals $$\frac{1}{d_i}\text{tr}(T_j)=<\chi_{\rho_i}, \chi_{\rho_j}>,$$ so the claim follows from Lemma~\ref{delta}. In particular, we find that $\Psi$ is surjective. Since the two sides of $\Psi$ have the same dimension by Proposition~\ref{A_W}, it follows that $\Psi$ is in fact an isomorphism. The remaining assertion is clear.
\end{proof}
\begin{defn}
    Let $G_F=Gal(\overline{F}/F)$ be the absolute Galois group of $F$. For any $\sigma\in G_F$ and any representation $\rho$, we denote by $\sigma(\rho)=\sigma\circ\rho$ which is also a representation and is irreducible when $\rho$ is so. This establishes an action of $G_F$ on $Irr(G)$. Let $$[Irr(G)]=G_F\setminus Irr(G).$$
    For any $\rho\in Irr(G)$, let $[\rho]=G_F\rho$ be its orbit under $G_F$, and we define $$e[\rho]=\sum\limits_{\rho'\in[\rho]}e(\rho').$$
\end{defn}
\begin{cor}
    $\{e[\rho]|[\rho]\in[Irr(G)]\}$ is the set of all primitive central idempotents of $F[G]$.
\end{cor}
\begin{proof}
    Let $e\in F[G]$ be a primitive central idempotent. The inclusion $F[G]\subseteq\overline{F}[G]$ identifies $F[G]$ as the subring of $\overline{F}[G]$ consisting of $G_F$-invariant elements. Then, we have $e=\sum\limits_ie(\rho_i)$ for some irreducible $\rho_i$. Since $e$ is $G_F$-invariant, the whole orbit of $e(\rho_i)$ appear in the sum whenever $e(\rho_i)$ itself does, so that $e=\sum\limits_ie[\rho_i]$ and the assertion follows as $e$ is primitive. 
\end{proof}

\section{Duals and hulls of group codes}
\subsection{Hulls of group codes}Let $\BF_q$ be a finite field of characteristic $p$. Recall that, if $C$ is a subspace of $\BF_q^n$ with dimension $k$, then we say that $C$ is a $q$-ary linear $[n,k]$-code, or simply, an $[n,k]_q$-code. Let $$<\ ,\ >:\BF_q^n\times\BF_q^n\rightarrow\BF_q$$ be the inner product on $\BF_q^n$ sending $u=(u_1, \dots, u_n)$ and $v=(v_1, \dots, v_n)$ to $u_1v_1+\dots+u_nv_n$. For any linear code $C\subseteq\BF_q^n$, we define $$C^{\perp}=\{v|<v,u>=0(\forall u\in C)\}$$ which is called as the dual code of $C$. It is easy to see that
\begin{enumerate}
    \item $\dim_{\BF_q}C+\dim_{\BF_q}C^{\perp}=n$;
    \item $(C^{\perp})^{\perp}=C$;
    \item $C_1\subseteq C_2$ if and only if $C_2^{\perp}\subseteq C_1^{\perp}$.
\end{enumerate}
\begin{defn}
    Let $G$ be a finite group. Then a 2-sided ideal of $\BF_q[G]$ is called a group code in $\BF_q[G]$.
\end{defn}
\begin{example}
    If $G=\BZ/n\BZ$ for some integer $n\ge1$, then $\BF_q[G]\cong\BF_q[x]/(x^n-1)$. Therefore a group code in such situation is a so-called cyclic code. 
\end{example}
For any representation $\rho:G\rightarrow GL_n(\overline{F}_q)$ we define $$\rho^*:G\rightarrow GL_n[\overline{\BF}_q],\ g\mapsto{^T\rho(g)^{-1}}$$
which is called as the contragredient of $\rho$. It is clear that if $\rho$ is irreducible, $\rho^*$ is also irreducible. For any $z=\sum\limits_{g\in G}f(g)g$ we define $$z^*=\sum\limits_{g\in G}f(g^{-1})g.$$
\begin{prop}
    Let $G$ be a finite group with $(|G|,p)=1$. Then, for any group code $C=(e)$ in $\BF_q[G]$, we have $C^{\perp}=(1-e^*)$. In particular, we have $h(C)=(e(1-e^*))$. 
\end{prop}
\begin{proof}
    Let us first consider the case when $C$ is minimal. Then, $e$ is primitive and hence of the form $e=e[\rho]$ where $\rho$ is an irreducible representation of $G$. Suppose that the $G_{\BF_q}$-orbit of $\rho$ consists of $\rho_1=\rho, \rho_2, \dots, \rho_s$. Then. $e=\sum\limits_{i=1}^se(\rho_i)$ which implies that 
    \begin{align*}
        e^*&=\sum\limits_{i=1}^se(\rho_i)^*\\&=\sum\limits_{i=1}^s\frac{d}{|G|}\sum\limits_{g\in G}\chi_{\rho_i}(g)g\\&=\sum\limits_{i=1}^s\frac{d}{|G|}\sum\limits_{g\in G}\chi_{\rho_i^*}(g^{-1})g\\&=\sum\limits_{i=1}^se(\rho_i^*)\\&=e[\rho^*].
    \end{align*}
    Here, $d$ is the dimension of $\rho$. It follows that 
    \begin{align*}
    <e,e^*>&=\sum\limits_{i,j=1}^s<e(\rho_i),e(\rho_j^*)>\\&=\sum\limits_{i,j=1}^s\frac{d^2}{|G|^2}\sum\limits_{g\in G}\chi_{\rho_i}(g)\chi_{\rho_j^*}(g)\\&=\sum\limits_{i,j=1}^s\frac{d^2}{|G|^2}\sum\limits_{g\in G}\chi_{\rho_i} 
    (g)\chi_{\rho_j}(g^{-1})\\&=\sum\limits_{i=1}^s\frac{d^2}{|G|}\\&=\frac{sd^2}{|G|}
    \end{align*}
    Since we also have $<e,1>=\sum\limits_{i=1}^s\frac{d}{|G|}\chi_{\rho_i}(e_G)=\frac{sd^2}{|G|}$, it follows that $1-e^*\in C^{\perp}$, i.e., $(1-e^*)\subseteq C^{\perp}$. Moreover, since $e^*=e[\rho^*]$ with $\rho^*$ is irreducible, we find that $1-e^*$ is the sum of all those primitive central idempotents except $e[\rho*]$, which implies that $(1-e^*)$ is maximal, so that $C^{\perp}=(1-e^*)$.
   
    Now we consider the general case so that $e=e_1+\dots+e_k$ is a sum of primitive central idempotents. Since distinct primitive central idempotents are orthogonal, it follows that 
    \begin{align*}
        C^{\perp}&=(e_1)^{\perp}\cap\dots\cap(e_k)^{\perp}\\&=(1-e_1^*)\cap\dots\cap(1-e_k^*)\\&=(\prod\limits_{i=1}^k(1-e_i^*))\\&=(1-e_1^*-\dots-e_k^*)\\&=(1-e^*).
    \end{align*}
\end{proof}
\begin{defn}
    For any $\rho\in Irr(G)$ let $s(\rho)$ be the number of elements in the $G_{\BF_q}$-orbit $[\rho]$ of $\rho$. If $d=\dim_{\overline{\BF}_q}\rho$ is the dimension of $\rho$, then we define $$d[\rho]=s(\rho)\times d^2.$$
\end{defn}
\begin{lem}
    For any $\rho\in Irr(G)$ we have $\dim_{\BF_q}(e[\rho])=d[\rho]$.
\end{lem}
\begin{proof}
    Let $(e[\rho])\otimes_{\BF_q}\overline{F}_q$ be the scalar extension of $(e[\rho])$ from $\BF_q$ to $\overline{F}_q$. It suffices to show that $dim_{\overline{\BF}_q}(e[\rho])\otimes_{\BF_q}\overline{F}_q=d[\rho]$. Since $(e[\rho])\otimes_{\BF_q}\overline{\BF}_q=(e(\rho_1))\oplus\dots\oplus(e(\rho_s))$, where $s=s(\rho)$ and $\{\rho_1=\rho, \dots, \rho_s\}$ is the $G_{\BF_q}$-orbit of $\rho$, we are reduced to show that $(e(\rho_i))$ is of dimension $d^2$ over $\overline{\BF}_q$. This follows from the Weddburn decomposition (1.2.1).
\end{proof}

\subsection{Proof of Theorem~\ref{main1}}

\begin{proof}
    For any $C=(e)$ with $e=\sum\limits_{[\rho]\in T}e[\rho]$ for some subset $T\subseteq [Irr(G)]$, we have $e^*=\sum\limits_{[\rho]\in T}e[\rho^*]$, which implies that $h(C)=(e(1-e^*))$ with 
    \begin{align*}
        e(1-e^*)&=\sum\limits_{[\rho]\in T}e[\rho]\sum\limits_{[\xi^*]\notin T}e[\xi]\\&=\sum\limits_{[\rho]\in T^*}e[\rho],
    \end{align*}
    where $T^*=\{[\rho]\in T|[\rho^*]\notin T\}$. Therefore, the necessity follows. Conversely, given such a subset $S$ satisfying (1) and (2), let $e=\sum\limits_{[\rho]\in S}e[\rho]$. Then, $e(1-e^*)=e$ so that $C=(e)$ is a group code in $\BF_q[G]$ with its hull of dimension k.
\end{proof}

\section{Hulls of dimensions $\leq3$}

In this section we apply Theorem~\ref{main1} to study group codes with hulls of dimension $\leq3$. In particular, we will see that if the group is non-abelian and simple then there exists no group codes with hull of dimension $\leq3$. Finally, we investigate the dihedral codes.

\subsection{One-dimensional hulls}In this subsection we give the proof of Theorem~\ref{main2} which yields a criterion for when there exist $1$-dimensional hull.
\\

\textbf{Proof of Theorem~\ref{main2}:} By Theorem 2.1 the necessary and sufficient condition is that there is a subset $S$ of $[Irr(G)]$ such that (a)$\sum\limits_{[\rho]\in S}d[\rho]=1$; (b)$[\rho^*]\notin S$ for any $[\rho]\in S$. However, (a) holds if and only if $S$ consists of a single orbit which contains a single $1$-dimensional representation $\rho:G\to\BF_q^*$, and (b) then means that $\rho^{-1}\neq\rho$. Therefore, we find that there exists a group code $C\subseteq\BF_q[G]$ with $dim_{\BF_q}h(C)=1$ if and only if there is a non-quadratic homomorphism from $G$ to $\BF^*_q$. Since such a homomorphism must factor through the abelianization $G^{ab}$ of $G$, the necessary and sufficient condition turns out to be the existence of non-quadratic homomorphism from $G^{ab}$ to $\BF_q^*$. Write $$G^{ab}=\BZ/n_1\BZ\oplus\dots\oplus\BZ/n_r\BZ$$ with $1<n_1|n_2|\dots|n_r$ so that $n_r=exp(G^{ab})$ is equal to the exponent of $G^{ab}$. Since $\BF_q^*$ is cyclic of order $q-1$, it follows that $$Hom(G^{ab},\BF_q^*)\cong\bigoplus_{i=1}^r\BZ/(n_i,q-1)\BZ.$$ Thus, there exists a non-quadratic homomorphism from $G^{ab}$ to $\BF_q^*$ if and only if $(n_i,q-1)\ge3$ for some $1<i\le r$, or equivalently, $(n_r,q-1)\ge3$.

\subsection{Two-dimensional Hulls}For the existence of group codes with $2$-dimensional hulls we have the following result.

\begin{thm}\label{hull2}
    Let $\BF_q$ be a finite field and $G$ a finite group with $(|G|,q)=1$. Write $G^{ab}=\BZ/n_1\BZ\oplus\dots\oplus\BZ/n_r\BZ$ with $1<n_1\mid n_2\mid\cdots\mid n_r$. Then there exists a group code $C$ in $\BF_q[G]$ with $\dim_{F_q}h(C)=2$ if and 
    only if one of the following conditions is satisfied:
    \begin{enumerate}
        \item $(n_r,q-1)\ge5$;
        \item $3\le(n_r,q-1)\le4$ and $2\le(n_{r-1},q-1)$;
        \item $(n_r,q^2-1)>(n_r,q-1)$ and $(n_r,q^2-1)>(n_r,q+1)$.
    \end{enumerate}
\end{thm}
\begin{proof}
    It follows from Theorem 2.1 that there exists a group code $C$ with a $2$-dimensional hull if and only if there is a subset $S\subseteq[Irr(G)]$ such that (a)$\sum\limits_{[\rho]\in S}d[\rho]=2$; (b)$[\rho^*]\notin S$ for any $[\rho]\in S$. It is natural to distinguish into the the following two cases:
    \begin{enumerate}
        \item[(A)] Let us first consider the case when $S$ consists of two orbits. Then, by condition (a), each must be a 1-dimensional representation with values in $\BF_q$. We denote them as $\rho_i:G\to\BF_q^*$ where $i=1,2$. Moreover, by condition (b), $\rho_1,\rho_2$ are not quadratic characters and we have $\rho_1^{-1}\neq\rho_2$. Since $$\Hom(G^{ab},\BF_q^*)\cong\bigoplus_{i=1}^r\BZ/(n_i,q-1)\BZ,$$
this is equivalent to the existence of a subset, also denote as $S$, in the right hand side of the above isomorphism consisting of two elements $x_1,x_2$ satisfying $2x_1,2x_2\neq0$ and $x_2\neq -x_1$. 

We claim that such an $S$ exists if and only if (1) or (2) holds. Indeed, we necessarily have $(n_r,q-1)\ge3$ as $2x_i\neq1$ for each $i=1,2$. If $(n_r,q-1)\ge5$, then it suffices to take $x_1=(0,\ldots,0,1)$ and $x_2=(0,\ldots,0,2)$. So, it remains to consider the situation when $3\le(n_r,q-1)\le4$. It is clear that in this situation we necessarily have $(n_{r-1},q-1)\ge2$. Conversely, if $(n_{r-1},q-1)\ge2$, then we may take $x_1=(0,\ldots,0,1)$ and $x_2=(0,\ldots,a,1)$ with $a$ any non-zero element in $\BZ/(n_r-1,q-1)\BZ$.
        \item[(B)]    We next consider the case when $S$ has a single orbit, represented by anon-quadratic 1-dimensional representation $\rho:G\to\BF_{q^2}^*$  such that $\Im(\rho)\nsubseteq\BF_q^*$ and $\rho^q\neq\rho^{-1}$. Since
$$\Hom(G^{ab},\BF_{q^2}^*)\cong\bigoplus_{i=1}^r\BZ/(n_i,q^2-1)\BZ,$$
this is equivalent to the existence of an element $x$ in the right hand side of the above isomorphism such that $2x,(q-1)x$ and $(q+1)x$ are all nonzero..

We claim that such an $x$ exists if and only if (3) holds. The necessity is clear, so it is enough to show that (3) is also sufficient for the existence of such an $x$. Since $\BZ/n\BZ[m]\simeq\BZ/(n,m)\BZ$ for any integers $n,m\geq1$, and a cyclic abelian group can have at most two $2$-torsion elements, it follows that we only need to show that
$$(n_r,q^2-1)>(n_r,q-1)+(n_r,q+1).$$
Indeed, under condition (3), we have $q\neq2$. By symmetry we may assume that $q\equiv3\pmod4$. If $\frac{(n_r,q^2-1)}{(n_r,q+1)}>2$, then $\frac{(n_r,q-1)}{(n_r,q^2-1)}+\frac{(n_r,q+1)}{(n_r,q^2-1)}\le\frac{1}{2}+\frac{1}{3}<1$. On the other hand, if $\frac{(n_r,q^2-1)}{(n_r,q+1)}=2$, then $2|n_r$ so that $\frac{(n_r,q^2-1)}{(n_r,q-1)}\ge4$ and $\frac{(n_r,q-1)}{(n_r,q^2-1)}+\frac{(n_r,q+1)}{(n_r,q^2-1)}\le\frac{1}{4}+\frac{1}{2}<1$, whence the above inequality and the assertion follows.
\end{enumerate}
\end{proof}

\subsection{Three dimensional hulls}For the existence of group codes with $2$-dimensional hulls we have the following result.

\begin{thm}\label{hull3}
    Let $\BF_q$ be a finite field and $G$ be a finite group with $(|G|,q)=1$.  Write $G^{ab}=\BZ/n_1\BZ\oplus\dots\oplus\BZ/n_r\BZ$ with $1<n_1\mid n_2\mid\cdots\mid n_r$. Then there exists a group code $C$ in $\BF_q[G]$ with $\dim_{F_q}h(C)=3$ if and 
    only if one of the following conditions is satisfied:
    \begin{enumerate}
        \item $(n_r,q-1)\ge7$
        \item $5\le(n_r,q-1)\le6$ and $(n_{r-1},q-1)\ge2$
        \item $3\le(n_r,q-1)\le4$ and $(n_{r-1},q-1)\ge3$
        \item $3\le(n_r,q-1)\le4$ and $(n_{r-1},q-1)=(n_{r-2},q-1)=2$
        \item $(n_r,q^2-1)>(n_r,q-1)\ge3$ and $(n_r,q^2-1)>(n_r,q+1)$
        \item $(n_r,q^3-1)>(n_r,q-1)$. 
    \end{enumerate}
\end{thm}
\begin{proof}
    It follows from Theorem 2.1 that there exists a group code $C$ with a $3$-dimensional hull if and only if there is a subset $S\subseteq[Irr(G)]$ such that (a)$\sum\limits_{[\rho]\in S}d[\rho]=3$; (b)$[\rho^*]\notin S$ for any $[\rho]\in S$. It is natural to distinguish into the the following three cases:
    \begin{enumerate}
        \item[(A)]Let us first consider the case when $S$ consists of three orbits. Then, by condition (a), each must be a 1-dimensional representation with values in $\BF_q$. We denote them by $\rho_i:G\to\BF_q^*$ where $i=1,2,3$. Moreover, by condition (b), $\rho_1,\rho_2,\rho_3$ are not quadratic characters and we have $\rho_i^{-1}\neq\rho_j$ for any $i\neq j$. Since $$\Hom(G^{ab},\BF_q^*)\cong\bigoplus_{i=1}^r\BZ/(n_i,q-1)\BZ,$$
this is equivalent to the existence of a subset, also denote as $S$, in the right hand side of the above isomorphism, consisting of two elements $x_1,x_2,x_x$ satisfying $2x_1,2x_2,2x_3\neq0$ and $x_i\neq -x_j$ for any $i\neq j$. 

We claim that such a subset $S$ exists if and only if one of (1), (2), (3) or (4) holds. Indeed, since $2x_i\neq0$, we necessarily have $(n_r,q-1)\ge3$. Therefore,
\begin{itemize}
    \item If $(n_r,q-1)\ge7$, then we may take $x_i=(0,\ldots,0,i)$ with $1\leq i\leq3$. 
    \item If $5\le(n_r,q-1)\le6$, then the last direct summand $\BZ/(n_r,q-1)\BZ$ yields only two such elements, say, $x_1=(0,\ldots,0,1)$ and $x_2=(0,\ldots,0,2)$. Thus, we must also have  $(n_{r-1},q-1)\ge2$.  Conversely, if $(n_{r-1},q-1)\ge2$, then we may take $x_3=(0,\ldots,0,a,1)$ for any nonzero $a$.
    \item If $3\le(n_r,q-1)\le4$, then the last direct summand $\BZ/(n_r,q-1)\BZ$ yields only one such elements, say, $x_1=(0,\ldots,0,1)$. Thus, we must also have  $(n_{r-1},q-1)\ge2$.  

Now, if $(n_{r-1},q-1)\ge3$, then we may take $x_2=(0,\ldots,0,a,1)$ and $x_3=(0\ldots,0,b,1)$ for distinct nonzero $a,b$. Finally, if $(n_{r-1},q-1)=2$, then we must have also $(n_{r-2},q-1)\geq2$, which is hence equal to $2$ as it is a divisor of $(n_{r-1},q-1)$. Conversely, if $(n_{r-2},q-1)=2$, then we may take $x_3=(0,\ldots,0,1,1,1)$.
\end{itemize}
        \item[(B)]We next consider the case where $S$ consists of two orbits. Then, by condition (a), one is a non-quadratic character $\chi:G\to\BF_q^*$, and the other is represented by a non-quadratic 1-dimensional representation $\psi:G\to\BF_{q^2}^*$ with $\psi^{q}\neq\psi^{\pm1}$. By Theorem~\ref{main2} and Theorem~\ref{hull2}, such an $S$ exists if and only if (5) holds.
        \item[(C)]Finally, we consider the case where $S$ has a single orbit, represented by a $1$-dimensional representation $\rho:G\to\BF_{q^3}^*$ such that $\Im(\rho_i)\nsubseteq\BF_q^*$ and $\rho^{-1}\neq\rho, \rho^q$ or $\rho^{q^2}$. Since $$\Hom(G^{ab},\BF_{q^3}^*)\cong\bigoplus_{i=1}^r\BZ/(n_i,q^3-1)\BZ,$$
this is equivalent to the existence of a subset, also denote as $S$, in the right hand side of the above isomorphism, consisting of  a single element $x$ satisfying $2x\neq0$ and $(q-1)x,(q^2-1)x$ are both nonzero. 

We claim that such an $S$ exists if and only if (6) holds. The necessity is clear. Conversely, if (6) holds. then $(n_r,q^3-1)\ge3(n_r,q-1)>(n_r,q-1)+1$ as $q^2+q+1$ is odd, which implies that there exists an $x\in\BZ/(n_r,q^3-1)\BZ$ such that $2x\neq0$ and $(q-1)x\neq0$. Since $(q^3-1,q+1)=2$, we must also have $(q+1)x\neq0$, whence the assertion.
    \end{enumerate}
    
\end{proof}

\subsection{An example-Dihedral case}
We end this section with a more concrete example, namely, dihedral group codes. Let $n$ be a positive integer and denote by $D_n=\BZ/n\BZ\rtimes\BZ/2\BZ$ the dihedral group generated by $t,s$ where $r$ (resp. $s$) has order $n$ (resp. $2$) such that $st=t^{-1}s$. Our result is as follows.

\begin{prop}
Let $\BF_q$ be a finite field and $n\geq1$ an integer such that $(n,q)=1$, then there exists no group code in $\BF_q[D_n]$ with hulls of dimension $\le8$.
\end{prop}
\begin{proof}
Since $d[\rho]=s(\rho)\cdot\dim_{\overline{\BF}_q}^2\rho$, we only need to prove that $[\rho]=[\rho^*]$ for any $\rho\in IrrG$ with dimension $\le2$. Since 
\begin{align*}
G^{ab}\simeq
\begin{cases}
\BZ/2\BZ,\text{ if }2\nmid n,\\
\BZ/2\BZ\oplus\BZ/2\BZ,\text{ if }2\mid n,    
\end{cases}
\end{align*}
$G^{ab}$ is annihilated by $2$ so that $\rho=\rho^{-1}$ for any $1$-dimensional $\rho$. Therefore, it remains to consider the case when $\dim_{\overline{\BF}_q}\rho=2$. Let $W$ be an invariant subspace of $\rho(r)$. Then, $sW$ is also invariant under the action of $r$ so that $\rho=W\oplus sW$ as it is irreducible. It follows that $\chi_{\rho}=\chi_{\rho^*}$, which implies that $\rho^{-1}$ is isomorphic to $\rho$.
\end{proof}

\end{document}